\documentclass[nonblindrev]{informs3}

\OneAndAHalfSpacedXI

\usepackage{natbib}
 \bibpunct[, ]{(}{)}{,}{a}{}{,}%
 \def\bibfont{\small}%

\usepackage{setspace}
\TheoremsNumberedThrough     
\ECRepeatTheorems

\EquationsNumberedThrough    

\usepackage{palatino}

\renewenvironment{proof}[1][Proof]{\textbf{#1.} }{\ \rule{0.5em}{0.5em}}
\begin{document}


\RUNAUTHOR{Buterin, Hitzig, and Weyl}

\RUNTITLE{A Flexible Design for Funding Public Goods}

\TITLE{A Flexible Design for Funding Public Goods}

\ARTICLEAUTHORS{%
\AUTHOR{Vitalik Buterin}
\AFF{Ethereum Foundation} 
\AUTHOR{Zo\"{e} Hitzig}
\AFF{Harvard University, \EMAIL{zhitzig@g.harvard.edu}} 
\AUTHOR{E. Glen Weyl}
\AFF{Microsoft Research, \EMAIL{glenweyl@microsoft.com}}
} 

\ABSTRACT{%
We propose a design for philanthropic or publicly-funded seeding to allow (near) optimal provision of a decentralized, self-organizing ecosystem of public goods.  The concept extends ideas from Quadratic Voting to a funding mechanism for endogenous community formation.  Citizens make public goods contributions to projects of value to them.  The amount received by the project is (proportional to) the square of the sum of the square roots of contributions received.  Under the ``standard model'' this mechanism yields first best public goods provision. Variations can limit the cost, help protect against collusion and aid coordination. We discuss applications to campaign finance, and highlight directions for future analysis and experimentation. 
}%


\maketitle

%


\section{Introduction}

In many contexts, a sponsor with capital wishes to stimulate and support the creation of public goods but is ill-informed about the appropriate goods to create. Thus, such a sponsor may want to delegate this allocation to a decentralized market process.  Examples of these contexts include campaign finance, funding open source software, public or charitable support for news media and the funding of intraurban public projects. Recent work on the theory of Quadratic Voting (henceforth QV; see \citealp{intro} for a survey) suggests that near-optimal collective decision-making may be feasible in practice, but relies on an assumption of a fixed set of communities and public\ goods that is inappropriate to this context.  In this paper we propose an extension of the logic of QV to this setting.

The basic problem we address can be seen by comparing two extreme ways of funding such a ecosystem, both of which are problematic.  On the one hand, a simple private contributory system famously leads to the under-provision of public goods that benefit many people because of the free-rider problem \citep{samuelsonpure}. The larger the number of people the benefit is split amongst, the greater the proportional under-provision. Conversely, a system based purely on membership or on some other one-person-one-vote (1p1v) system cannot reflect how important various goods are to individuals and will tend to suppress smaller organizations of great value.  We aim to create a system that is as flexible and responsive as the market, but avoids free-rider problems.\par

Our solution is to modify the funding principle underlying the market to make it nonlinear.  In a standard linear private market, the funding received by a provider is the sum of the contributions made by the funders.  In our ``Quadratic Finance'' (QF) mechanism, the funding received by a provider is the square of the sum of the square roots of the contributions made by the funders.  Holding fixed contribution amounts, funding thus grows with the square of the number of members.  However, small contributions are heavily subsidized (as these are the most likely to be distorted by free-riding incentives) while large ones are least subsidized, as these are more like private goods. Under the standard selfish, independent, private values, quasi-linear utility framework, our mechanism leads to the utilitarian optimal provision of a self-organizing ecosystem of public goods. 

Existing systems such as matching funds for infrastructure projects, political campaigns, charitable contributions, and other public goods aim to capture similar benefits, but do so in an unsystematic way. For example, a variety of public goods are funded through matching programs, whereby an institutional body (a government, corporation, political party, etc.) matches individual contributions either 1:1 or in some other ratio. For example, New York City matches small contributions to campaigns for elected office (city council, mayor, comptroller, public advocate), matching contributions 6:1 and up to \$175. Many corporations use similar rules: one of our employers matches charitable contributions by all full-time employees up to \$15,000 a year. Doing so amplifies small contributions, incents more contributions and greater diversity in potential contributors, and confers a greater degree of influence on stakeholders in determining ultimate funding allocations. 

Matching programs are not only common in public and charitable funding, but also follow an intuitive logic that has built a variety of public policies. Indeed, the very idea of tax deductibility for charitable contributions is a form of governmental matching. But while matching funds share the spirit of our funding principle, they lack a systematic design. Funding ratios and match thresholds are often set in largely arbitrary ways. The QF mechanism can be seen as offering a coherent design that captures the central motivation of matching funds in a mechanism that is (approximately) optimal from the perspective of economic theory. 

We begin the paper in Section \ref{background} by providing background on the economic theory of public goods. We then develop a simple but general mathematical model in Section \ref{model} of public good provision and use it to illustrate the failures of both private contributory systems \citep{varian} and 1p1v \citep{bowen}. Then in Section \ref{analysis}, we describe QF formally in a simple model. We  show mathematically that QF leads to optimal public goods provision. We then turn, in Section \ref{extensions}, to variations and extensions that enrich our understanding of QF and its range of applications. 

Having developed this apparatus, we change gears in Section \ref{apps} to describing one application of QF in detail, to campaign finance reform. We also briefly describe how QF may be applied to open source software ecosystems, news media finance, charitable giving and urban public projects, and how QF aligns with qualitative features of previous solutions while more smoothly covering a wider range of cases and problems. We conclude the paper in Section \ref{conclusion} with a discussion of directions for future research and experimentation.

\section{Background}\label{background}

\vspace{\baselineskip}

\subsection{Public Goods Problems}
One of the most fundamental problems in political economy is variously known as the ``free-rider'', ``collective action'' or ``public goods'' problem; we will use the term ``public goods''.  All these terms refer to situations in which individuals can or do receive benefits from shared resources and investments that may be more valuable than the contribution they individually make to those resources. These goods cannot be efficiently priced due to the expense or inefficiency required to exclude individuals from access. By ``public good" we refer to any activity with increasing returns in the sense that the socially efficient price to charge for the activity (marginal cost) is significantly below the average cost of creating the good.

Seen in this broader light, public goods are at the core of human flourishing. Civil society sustains itself precisely because the whole is greater than the sum of its parts.  Contemporary economic thought has increasingly emphasized the centrality of increasing returns, especially through investment in innovation and knowledge, to development, beginning with the work of \citet{romer1986}.  As the exploding literature on agglomeration and spatial economics emphasizes, the cities that created the idea of middle classes (viz. {\em bourgeoisie}) and the citizen could not exist without increasing returns \citep{krugman}.  Yet, despite this centrality,  classical capitalism deals poorly with such activities. Because each individual, if she acts selfishly, only accounts for the benefits she receives and not the benefits to all other individuals, funding levels will not scale with the number of individual beneficiaries as would be desirable.

A range of institutions have emerged to address the public goods problem in modern society. The most canonical and perhaps the most important institution that coordinates provision of public goods is the contemporary democratic nation state.  Such states use taxation and voting-based governance systems to determine which public goods should be provided. The other most prevalent method for addressing public goods involves converting them to private goods by imposing technologies (e.g. walls, fee collectors at parks, digital rights management for information, etc.) that allow individuals to be excluded.  Other institutions draw on moral, cultural, religious or social motives to induce individuals to contribute to charitable providers of public goods. Some intermediate institutions mix elements of these three ideal types.\footnote{One example of an institution that mixes these ideal types is a local government with some ability to exclude, which citizens can move across at some cost and to which they have some loyalty (and thus often donate their time). Another example is an exclusive but not-for profit club.}

Unfortunately, all of these institutions have limitations. 1p1v democratic systems (even when they work appropriately) respond to the will of the majority, not necessarily to what would create the greatest overall value.  They often oppress minorities or are subverted by minorities to avoid such oppression.  They are also costly to set up, rigid and do not easily adapt to demands for different and new levels of organization. Private (usually corporate) exclusion-based efforts are, while more flexible, usually cumbersome and costly to impose, often lack effective feedback mechanisms that ensure they serve the interests of their members and, perhaps most importantly, inefficiently exclude potential users.  Charitable organizations are often more responsive and flexible than either of the other forms, but they rely on motives that can be difficult to closely align reliably with the common good outside of the relatively small groups in which they are often very effective \citep{ostrom}.  Outside such groups the often instead get captured by status motivations and parochial, even exclusionary, interests.\footnote{See \citet{reich} for a discussion of these issues in the context of contemporary American capitalism.}

\subsection{Literature}\label{literature}

\citet{clarke} and \citet{groves}, recasting the insights of \citet{vickery}, proposed a solution to the collective action problem, in the form of a  mechanism for individuals to reveal their preferences over public goods to a government or other central clearinghouse to overcome the rigidity and inefficiency of majority rule. This system, known as the VCG mechanism, was shown in \citet{greenlaffont1977, greenlaffont1979} to be the only dominant strategy incentive-compatible mechanism for producing public goods. The system is fragile to collusion and risky for participants. Though \citet{smith1980} showed that variants of VCG succeed in laboratory experiments, others have concluded that VCG mechanisms are generally impractical \citep{impracticalVCG}.

Other mechanisms for near-optimal collective decision-making have been proposed. \citet{gl} and \citet{hz} both suggested a quadratic mechanism for determining the level of continuous public goods. But, their methods require either a centralized iterative process or depend heavily on a strong assumption of complete information, and they do not, in general, satisfy individual rationality. The basic insight of quadratic pricing of collective choices reemerged recently in \citet{qvb}'s proposal for what he called ``Quadratic Vote Buying''. In particular, he proposed allowing individuals to buy votes, paying the square of the votes they buy.  He argued, and \citet{qveq} proved, that under standard assumptions (similar to those we use below) in large populations QV leads to approximately optimal decisions on public goods. 

However, while QV addresses the inefficiency of standard 1p1v voting systems for a \textit{given set of decisions and collectives,} it doesn't solve the problem of flexibility.  That is, it does not allow the set of public goods to emerge from a society organically, and effectively assumes a previously-specified organizational structure that has to be taken as an assumption or imposed by an authority.  In this paper we extend the ideas around QV to address these limitations. In sum, our mechanism improves on existing mechanisms by allowing for greater flexibility in the set of goods to be funded, by eliminating the assumption of complete information and by satisfying individual rationality constraints.

\section{Model}\label{model}

\vspace{\baselineskip}
We develop a flexible model for a society choosing which public goods to fund. Consider a society of  \( N \) citizens  \( i=1, \ldots N \). We assume throughout what follows that we can verifiably distinguish among and identify these citizens, though we will discuss the possibility that they may collude (see \ref{collusion} below). We use the term ``society'' to refer to the set of all participants and the word ``community'' to refer to groups that fund a particular public good; however, in many applications, the relevant ``society'' is itself a community within a broader setting.

There is a set of potential public goods  \( P \). We do not make any assumption about the nature of this set. There may be measure theoretic issues for some cardinalities of the set, but we will ignore these issues in our first presentation of the basic idea. In particular, there is no sense in which the set of public goods need be specified externally or in advance; any citizen may at any time propose a new public good. We denote a typical public good  \( p \in P \).

\subsection{Individual preferences and actions}

Let  \( V_{i}^{p} \left( F^{p} \right)  \) be the currency-equivalent utility citizen  \( i \) receives if the funding level of public good  \( p \) is  \( F^{p} \).  We assume all public goods generate independent value to citizens (no interactions across public goods) and that citizens have quasi-linear utility denominated in units of currency. We also assume a setting of complete information, though given the flexible set up of the problem, our results do not rely heavily on this assumption. We also abstract away from issues about observability and timing of contributions.

Our interest here is in maximization of dollar-equivalent value rather than achieving an equitable distribution of value (we assume that an equitable distribution of basic resources has been achieved in some other manner, such as an equal initial distribution of resources). For purposes of simplifying the analysis below, we assume all functions  \( V_{i}^{p}\) are concave, smooth and increasing.  Absent these assumptions some complications may arise (as we return to in \ref{concavity}) but it is easier to abstract from them in presenting the core ideas.

Each citizen  \( i \) can make contributions to the funding of each public good  \( p \) out of her personal resources  \( c_{i}^{p} \). The total utility of citizen  \( i \) is then
\begin{equation}
\sum _{p}^{}V_{i}^{p} \left( F^{p} \right) -c_{i}^{p}-t_{i}
\end{equation}    
where  \( t_{i} \) is a tax imposed on individual  \( i \). In this framework, different funding mechanisms for public goods are different formulae for relating  \(  \left \{ F^{p} \right\}_{p \in P} \) to  \(  \left \{ c_{i}^{p} \right\}_{i \in I, p \in P} \), with any surplus or deficit being made up for by taxes that do not influence behavior.

\subsection{Funding mechanisms}

A funding mechanism in our flexible public goods setting defines the total amount of funding received for each good in the set $P$, given all individual contributions $c_i^p$. Formally, a mechanism is a mapping from the set of all individual contributions to funding levels for all goods. The set of individual contributions is comprised of vectors $\mathbf{c}^p = (c_1^p, c_2^p, ..., c_N^p)$ where subscripts index citizens. Thus, $\mathbf{c}^p$ is a vector in $\mathbb{R}^N$, and we denote by $\mathcal{C}^{|P|}$ the space of all possible collections of funding levels for each good $p$ given contributions from the $N$ citizens, i.e. $\{\mathbf{c}^p\}_{p\in P}$. The set of all final funding levels for all goods $p\in P$ is the set $\mathcal{F}$, which has $|P|$ real-valued elements $\mathbf{F} = (F^{1}, F^2, ..., F^{|P|})$,  with $F^p \in \mathbb{R}. $

\begin{definition}[Funding Mechanisms] A \textit{funding mechanism} \( \Phi: \mathcal{C}^{|P|} \to \mathcal{F}\)  determines the total level of funding for each good $p\in P$, such that $\Phi(c_i^p) = \{F^p\}_{p \in P}$.\footnote{In a slight abuse of notation, we will sometimes use $\Phi$ to refer to the subcomponent $\Phi^p$ which maps individual contributions to good $p$ into funding levels $F^p$ for that particular good $p$.}
\end{definition}
Budget balance requires that  \(\sum _{i}^{}t_{i}= \sum _{p}^{} \left( F^{p}- \sum _{i}^{}c_{i}^{p} \right)  \), i.e. taxes make up for any deficit between individual contributions and total funding levels. Before studying such mechanisms, however, we consider what social welfare maximization requires.  Our analysis here is the special case of \citeauthor{samuelsonpure}'s analysis in the case of quasi-linear utility.

\subsection{Welfare and optimality}
Given the simple set up of our model, welfare calculations are straightforward. Total social welfare is 
\begin{equation}
\sum _{p}^{} \left(  \sum _{i}^{}V_{i}^{p} \left( F^{p} \right)  \right) -F^{p}  
\end{equation}  
by the budget constraint. Let  \( V^{p} \left( F^{p} \right)  \equiv  \sum _{i}^{}V_{i}^{p} \left( F^{p} \right)  \) be the total value all citizens derive from the good.

Maximizing $V^p(F^p)$ over all weakly positive funding levels  \(  \left \{ F^{p} \right\}_{p \in P} \) for all goods, given concavity and smoothness of the  \( V \) functions gives a simple solution:   \( F^{p} \) is 0 if  \( V^{p\prime} \left( 0 \right)  \leq 1 \); \( F^{p} \)  takes on the unique value satisfying \( V^{p\prime}=1\) otherwise. That is, the total marginal value derived from the good should equal  \( 1 \).  
 
\begin{definition}[Optimality] A funding mechanism $\Phi$ is optimal if for all $p \in P$: (i) $V^{p\prime}(0) \leq 1$ implies $F^p = 0$, and (ii) $V^{p\prime}(0) > 1$ implies $V^{p\prime}=1$. 
 \end{definition}

\subsection{Suboptimal mechanisms}

We now consider two suboptimal funding mechanisms. The first, which we refer to as ``private contributions'', has the total contributions exactly equal to the sum of individual contributions, as analyzed in \citet{varian}. There is no centralized funding based on individual contributions, and thus no need for taxes or transfers. 
\begin{definition}[Private Contributions Mechanism] Under \textit{private contributions}, \[ \{F^{p}\}_{p\in P}= \Phi^{priv}(c_i^p) = \left\{\sum _{i}^{}c_{i}^{p}\right\}_{p \in P}. \] 
\end{definition}
Note that  \( t_{i}=0 \) under private contributions. This mechanism corresponds to the traditional formula used for charitable giving; while there are sometimes public matching funds that linearly scale contributions, these will not greatly change our conclusions, which closely follow the analysis of \citeauthor{varian}\footnote{An extension that accounts for linear matching funds would simply add a scaling factor to private contributions. This mechanism for linearly scaled financing, which we might call $\Phi^{LF}$ has total funding defined by $\Phi^{LF}(c_i^p) = \{\sum_i \alpha c_i^p\}_{p \in P}$, with $\alpha > 1$. This scaling factor clearly does not change our analysis, it simply scales the degree of underfunding.} In this case, every citizen  \( i \) seeks to maximize, in determining her contribution to good  \( p \)
\begin{equation}
 V_{i}^{p} \left(  \sum _{j}^{}c_{j}^{p} \right) -c_{i}^{p}. 
\end{equation}

\begin{proposition}[Suboptimality of Private Contributions] The private contributions mechanism $\Phi^{priv}$ is suboptimal.
\end{proposition}

\begin{proof}
Maximization requires (differentiating) that for any citizen  \( i \) making a positive contribution to good  \( p \) that 
 \[ V_{i}^{p\prime} \left( F^{p} \right) =1. \] 
That is, the level funding must be such that \textit{a single citizen's} marginal value equals  \( 1 \). Summing across citizens, $V^{p\prime} =1$ only when $c_i^p>0$ for a single $i$, and $c_j^p=0$ for all $j\neq i$. When there is more than one contribution to good $p$, generically $V^{p\prime} > 1$. 
\end{proof}

 If a large set of citizens benefit significantly from a public good, this will typically lead to severe underfunding.  For example, if all citizens are homogeneous, this is equivalent to  \( V^{p\prime}=N \), or setting the total marginal utility of the good to  \( N \) times the level it should be at. When citizens have heterogeneous preferences, matters are even worse, at least from a distributive perspective: only the single citizen who cares most on the margin about the good has any influence on its provision. Matters are more pessimistic yet if citizens can make negative contributions (privatize public goods), as then the lowest valuation citizen determines the provision level.\footnote{\citet{aa} suggest a system that sounds superficially different from purely private contributions but will typically lead to similar results.  They suggest every citizen be compelled to give\ some\ fixed amount to public goods (in fact, they suggest funding this using progressive taxes, but from the efficiency perspective we take here these are basically equivalent).  If there is a constrained set of public goods, this may have some impact in raising overall funding levels, but will not move things much towards optimality.  But if there is a sufficiently rich set of goods, such that each individual has a good that is equivalent to giving the money back to herself, this yields just the same result as capitalism: every individual uses the money to pay herself back, unless she has the greatest value for the public good. }

Another mechanism, which we will call ``1p1v'' works as follows.  Majority voting determines whether to fund each public good, and the goods selected receive funding through taxes and transfers.

 \begin{definition}[1p1v Mechanism] The \textit{1p1v} Mechanism $\Phi^{1p1v}$ satisfies \[\{F^p\}_{p\in P} = \Phi^{1p1v}(c_i^p) = \{N\cdot [Median_{i} V_{i}^{p\prime} \left( F^{P} \right) =1]\}_{p \in P}.\]
 \end{definition}
Clearly 1p1v does not lead to optimality, as the mean must be used in the above formula rather than the median to recover  \( V^{p\prime}=1 \), \ as \citet{bowen} observed.  \citet{bergstrombowen1} discussed the situations under which the mean is likely to be a good approximation for the median, and demonstrated the generic inefficiency of 1p1v type systems. 

\begin{proposition}[Suboptimality of 1p1v ] The 1p1v Mechanism $\Phi^{1p1v}$ does not guarantee optimal funding levels. 
\end{proposition}
\begin{proof}
In order for $\Phi^{1p1v}$ to recover optimal funding levels, it must be that $\forall p \in P$, 
\begin{equation}\label{equal}
Median_{i} V_{i}^{p\prime} \left( F^{P} \right)=\frac{1}{N}\sum_{i=1}^N V_{i}^{p\prime} \left( F^{P} \right).
\end{equation}
The fact that condition \eqref{equal}  implies efficient funding follows from quasilinear utility (so that $V_i^{p\prime}(F^p)$ is monotone and decreasing). While there are some cases in which \eqref{equal} will hold, as discussed in \citet{bowen}, it may be that 
\begin{equation}\label{greater}
Median_{i} V_{i}^{p\prime} \left( F^{P} \right)>\frac{1}{N}\sum_{i=1}^N V_{i}^{p\prime} \left( F^{P} \right)
\end{equation}
or 
\begin{equation}\label{lessthan}
Median_{i} V_{i}^{p\prime} \left( F^{P} \right)<\frac{1}{N}\sum_{i=1}^N V_{i}^{p\prime} \left( F^{P} \right)
\end{equation}
and thus, generically, $\Phi^{1p1v}$ is not always efficient. Depending on whether \eqref{equal}, \eqref{greater}, or \eqref{lessthan} holds, it may be that: (i) $V^{p\prime}=1$, (ii) $V^{p\prime}<1$, or (iii) $V^{p\prime}>1$. That is $\Phi^{1p1v}$ may recover optimal funding levels, or lead to over or under funding on the margin. 
\end{proof}

Public good funding levels will tend to be higher and probably more accurate than under purely private contribution schemes, which is likely why most developed countries use democratic mechanisms for determining levels of public goods. However, clearly the median is often a poor approximation for the mean, especially for goods of value to smaller communities or for ``entrepreneurial public goods,'' the value of which is not widely understood at the time of funding. Such goods may well receive no funding from 1p1v---indeed, in practice small communities and entrepreneurial public goods are often funded primarily by charity or other private contributory schemes rather than 1p1v.\par

Some improvements on 1p1v are possible. \citet{bergstrombowen2, bergstrombowen1} argued that if there is some proxy for which citizens will benefit most from a good and we can tax them for it, 1p1v systems yield better outcomes. In such settings, everyone will agree on whether a given good is desirable. But this result begs the question in an important sense: if we know how much citizens benefit from a good, then any consensual mechanism will work well. Our goal is to find appropriate funding level \textit{without} assuming such prior centralized knowledge. 

\section{Design and Analysis} \label{analysis}

Consider the funding mechanism, which we refer to as the \textit{Quadratic Finance} (henceforth \textit{QF}) mechanism.
\begin{definition}[Quadratic Finance Mechanism] The Quadratic Finance Mechanism satisfies $$\{F^{p}\}_{p\in P}= \Phi^{QF}(c_i^p) = \left\{ \left(  \sum _{i}^{}\sqrt{c_{i}^{p}} \right) ^{2}\right\}_{p \in P}.$$
\end{definition} 
For the moment, assume $\Phi^{QF}$ is funded by the deficit
\begin{equation}
\sum _{p}^{} \left[  \left(  \sum _{i}^{}\sqrt{c_{i}^{p}} \right) ^{2}- \sum _{i}^{}c_{i} \right] 
\end{equation} 
being financed by a per-capita tax on each citizen. We also will, for the moment, assume that citizens ignore their impact on the budget and costs imposed by it.  Whether this is an innocuous assumption will depend on context as we discuss further in \ref{deficit}.\footnote{Indeed, it is easy to see that for goods with widespread support, the individual contributions will only supply a small fraction of total funding. For these goods, the mechanism serves as a mode of eliciting preferences more than anything else, and thus a per-capita tax may be problematic for goods with highly skewed benefits.} Many potential applications of this mechanism will involve dedicated government funding or funding from philanthropy, in which case citizens' impact on the deficit may matter much less (if at all). Nevertheless, for now it is easiest to understand the logic of the mechanism without worrying about the deficit. 

\subsection{Baseline analysis}\label{baseline}

Under this assumption, citizen \( i \)'s contribution to good  \( p \) will be chosen to maximize
\begin{equation}
V_{i}^{p} \left(  \left(  \sum _{j}^{}\sqrt{c_{j}^{p}} \right) ^{2} \right) -c_{i}^{p}. 
\end{equation}

Any positive contribution will thus have to satisfy
\begin{equation}\label{lrcontributions}
 \frac{2V_{i}^{p\prime} \left( F^{p} \right)  \left(  \sum _{j}^{}\sqrt{c_{j}^{p}} \right) }{2\sqrt{c_{i}^{p}}}=1 \leftrightarrow V_{i}^{p\prime} \left( F^{p} \right) =\frac{\sqrt{c_{i}^{p}}}{ \sum _{j}^{}\sqrt{c_{j}^{p}}} 
\end{equation}
by\ differentiation. 
\begin{proposition}[Optimality of Quadratic Finance] The Quadratic Finance mechanism $\Phi^{QF}$ guarantees optimal funding levels.
\end{proposition}
\begin{proof}
Adding the expression in \eqref{lrcontributions} across citizens yields \(V^{p\prime} \left( F^{p} \right) =1\). Thus, $\Phi^{QF}$ satisfies optimality. 
\end{proof}

It is easy to check that the conditions for any positive contribution being made are also optimal (viz. precisely when  \( V^{p\prime}>1 \)).

\subsection{Intuition}

We briefly discuss an alternative derivation of the QF rule to provide further intuition and insight into the logic of the mechanism. This derivation translates an appealing normative property of a generic solution to a public goods problem into a differential equation, and shows that QF is its solution. 

Each citizen has a ``degree of contribution'' to a collective good that is a function of how much she gives: $h\left(c_j\right)$ for some scalar function $h$. These contributions are at least quasi-additive across citizens so the total amount of funding is $g\left(\sum_i h\left(c_i\right)\right)$ for some scalar function $h$. How should citizens choose their degree of contribution?  One appealing normative property to counter free-riding might be: individuals should not act in line with solely self-serving motives.

How might this normative property govern behavior in public goods provision?\footnote{This normative property is closely related to a principle in moral philosophy famously formalized by \citet{kant} as the ``categorical imperative":  ``act only according to that maxim whereby you can, at the same time, will that it should become a universal law". Relatedly, \citet{roemer} has suggested that the right solution to the public good problem is to induce a change in human behavior so that every citizen acts according to a ``Kantian equilibrium.''} The standard logic of free-riding is that each citizen imagines that she would be willing to contribute to a public good if, by her doing so, everyone else would as well. For example, each citizen might be willing to see her taxes increase by 1\% to fund a public good, but would be unwilling to contribute unilaterally. 

Thus following this logic, it may be desirable to have a mechanism such that a citizen $j$ could, by increasing $h\left(c_j\right)$ by 1\% see funding increase by 1\% of $\sum_i h\left(c_i\right)$. Such a rule can be represented by a simple ordinary differential equation.  Namely, for each $j$ we want
\begin{equation}\label{CI}
\frac{\partial g\left(\sum_i h\left(c_i\right)\right)}{\partial c_j}=\frac{\sum_i h\left(c_i\right)}{h\left(c_j\right)}.
\end{equation}
This differential equation directly implies QF. To see this, note that
\begin{equation}
\frac{\partial g\left(\sum_i h\left(c_i\right)\right)}{\partial c_j}=g'\left(\sum_i h\left(c_i\right)\right)h'\left(c_j\right)
\end{equation}
so that \eqref{CI} becomes
\begin{equation}
g'\left(\sum_i h\left(c_i\right)\right)h'\left(c_j\right)=\frac{\sum_i h\left(c_i\right)}{h\left(c_j\right)}.
\end{equation}
Structurally, the $g'$ term must treat all elements in the sum of $h$'s symmetrically and the $h'$ term must only include $c_j$.  Thus we must have that 
$$g'\left(\sum_i h\left(c_i\right)\right)=k\sum_i h\left(c_i\right)\iff g'(x)=kx$$
for some constant $k$ and
$$h'\left(c_j\right)=\frac{1}{k h\left(c_j\right)}\iff h'(x)=\frac{1}{kh(x)}.$$
Integrating these we obtain that $g(x)=\frac{k}{2}x^2+m$ and $h(x)=\frac{2\sqrt{x}}{k}+n$.  If we want the funding of a project with no contributions to be $0$, $m$ and $n$ should both be $0$, narrowing our solution to $g(x)=\frac{k}{2}x^2$ and $h(x)=2\frac{\sqrt{x}}{k}$.  If we want a mechanism in which a good with a single contributor is funded as in private contributory schemes, we obtain $k=2$ and thus QF.

\subsection{Properties of the Quadratic Finance mechanism}\label{properties}

This discussion leads us naturally to a consideration of the properties of the QF mechanism. 

First, QF is homogeneous of degree one in the sense that if a fixed set of citizens are contributing and double their contributions, the funding also doubles. Homogeneity of degree one is a useful and reassuring property, as it implies:

\begin{itemize}
	\item Changing currencies makes no difference to the mechanism.

	\item Groups can gain nothing by splitting or combining projects with the same group of participants.

	\item It matters little precisely how frequently the mechanism is run (whether donations are aggregated at the monthly, daily or yearly level) unless the pattern of donations is temporally uneven in a systematic way.
\end{itemize}

Second, consider the case in which every contributing citizen makes an equal contribution, say of one unit, as we vary the number of citizens contributing  \( N_{c} \).  In this case, the funding received is  \( N_{c}^{2} \). Thus, holding fixed the amount of the contribution, the funding received grows as the square of the community size.  This property is also  intuitive and reassuring, as we saw above that under purely private contributions, there is a factor  \( N_{c} \) underfunding of goods on the margin.  It is thus natural to solve this underfunding by scaling funding levels by the community size.

Third, and relatedly, note that a community that splits in half with roughly similar contribution profiles will receive half the aggregate funding of the total community: both halves will receive one quarter.  This feature of the mechanism is a clear deterrent against fragmentation and atomization, and is the core reason why the QF mechanism can solve the public goods problem.  However, this feature does not at all imply that under QF only extremely large communities will form. Different collections of citizens will have different purposes in using their funds, some in smaller groups and some in larger ones. 

The trade-off between preference heterogeneity and the benefits of scale is well-known to political economists. For example, this trade-off is discussed in the literature on the optimal size of nations \citep{alesina}. QF does not prejudge the optimal size of collectives, but unlike purely private contributory schemes or 1p1v offers a mechanism that creates neutral incentives among social organization of different sizes. This feature turns out, however, to require much greater funding \textit{for a given contribution profile} to larger grouping for the obvious collective action reason (see below): each citizen will tend to contribute less, absent this incentive, to larger groupings where she receives a smaller share of relevant benefits.

Fourth, note that the mechanism reverts to a standard private good in the case that a single citizen attempts to use the mechanism for her own enrichment.  In cases in which the overwhelming bulk of contributions come from one citizen, other contributions to the sum of square roots approximately drop out and we are left with the square of the square root, which is simply the contribution itself.  More broadly, as we approach the limit in which goods are private, the mechanism treats the contributions as contributions to a private good.

Fifth, and really just to summarize, the mechanism provides much greater funding to many small contributions than to a few large ones.  This feature does not result from imposing an external principle of equity or distributive justice, though there may be good reasons from those perspectives to admire the outcome it delivers. It instead results directly from the logic of the mechanism. This logic aligns well with a central concern in democratic theory since at least \citet{madison} and famously associated with Mancur Olson's \citeyearpar{olson} \textit{Logic of Collective Action}: Large communities of citizens that each receive only a small benefit tend to be disadvantaged by private contributory schemes relative to concentrated interests.

It is useful to highlight the more primitive attributes of our design that underlie the above properties by comparing QF to other efficient mechanisms for funding public goods. As mentioned in \ref{literature}, QF is a marked improvement over other proposals in certain ways, though a more rigorous comparison is left for future work. In particular, under QF, individuals have \textit{private information}---there is no need for complete information or a centralized process as in \citet{gl} and \citet{hz}. In addition, unlike \citet{gl}, QF respects \textit{individual rationality} constraints, which is the key to solving the free-rider problem. Any citizen with a positive valuation of a particular good has incentives to contribute. QF is also flexible in the sense that it does not require an enumeration of all participants, and does not rely on preconceived estimates of the number of individuals who benefit from particular goods. Contrast this flexibility with the mechanism proposed in \citet{falkinger} which achieves individual rationality by rewarding and penalizing deviations from the average contribution in order to achieve efficient funding levels---such a mechanism is obviously problematic for goods of value to small communities, and requires a full specification of the set of taxable citizens. 

Some of these properties may make a QF system vulnerable to collusion or manipulation, as we will return to in \ref{collusion}. But, overall we view these properties as heartening confirmations that our analysis addresses a wide range of issues relevant to public goods problems.

\subsection{User interface}

Precisely what the QF mechanism would ``look like'' is beyond our scope here, but a brief description of a possibility will hopefully help readers imagine how it might be feasible. Any citizen could at any time propose a new organization to be included in the system. Depending on the context, there might be a more or less extensive process of being approved to be listed in the system by an administrator; this approval process would be especially important for a philanthropically-sponsored implementation, as the philanthropist is unlikely to be willing to fund just any project. 

Citizens could contribute their funds towards (or possibly against, see \ref{negative} below) any listed project at some regular interval, such as monthly.  Citizens would be given some (possibly imperfect and delayed, for security purposes) indication of the total funding level of various projects.  Such information would help citizens determine: (i) the amount of funding projects would receive if they contributed a bit extra (likely aided by appropriate visualizations and ``calculators'') to a particular project; and (ii) whether a project has enough funding to be successful. The dissemination of funding information would help avoid fragmentation. Given the far greater funding that a project supported by many can receive as compared to a project with a few supporters, there would be far less incentive than under private contributions for a thousand projects to proliferate.

As we discuss in \ref{collusion} and \ref{concavity} below, various more detailed features of the system would be needed to help ensure security and enable coordination among participants.  Furthermore, the precise look and feel of the system requires much more thought and even might affect the formal rules in some way. 

\subsection{Incorporating the deficit}\label{deficit}

In the preceding analysis, we assumed that citizens ignore their impact on the deficit for clarity.  We will now see how the elimination of this assumption may alter our results. 

Suppose that citizen  \( i \) has a shadow value of  \(  \lambda _{i}\) on reducing the budget deficit. We can think of $\lambda_i$ as the fraction of the deficit that will be funded by taxing citizen $i$. Alternatively, as we will explore in \ref{budgeted} below, $\lambda_i$ can be interpreted as the cost to citizen $i$ of reduced funding of other public goods that a greater deficit will require. 
\begin{definition}[Aggregate Cost of Deficit]
The aggregate cost of an increased deficit is  \(  \Lambda  \equiv  \sum _{i}^{} \lambda _{i} \).
\end{definition}
We assume that \(  \lambda _{i} \) is on the order of  \( \frac{1}{N} \), so that the aggregate cost of an increased deficit, $\Lambda$ is around 1. Under these assumptions, in a large society no citizen is financing a large share of the deficit. Citizen  \( i \) seeks to maximize in her contributions to project  \( p \)
\begin{equation}
 V_{i}^{p} \left(  \left(  \sum _{j}^{}\sqrt{c_{j}^{p}} \right) ^{2} \right) -c_{i}^{p}- \lambda _{i} \left(  \left(  \sum _{i}^{}\sqrt{c_{i}^{p}} \right) ^{2}- \sum _{i}^{}c_{i} \right) .
\end{equation}
The associated first-order condition for maximization is
\begin{equation}\label{deficitfoc}
\frac{2 \left[ V_{i}^{p\prime} \left( F^{p} \right) - \lambda _{i} \right]  \left(  \sum _{j}^{}\sqrt{c_{j}^{p}} \right) }{2 \sqrt{c_{i}^{p}}}=1- \lambda _{i} \leftrightarrow V_{i}^{p\prime} \left( F^{p} \right) - \lambda _{i}=\frac{\sqrt{c_{i}^{p}}}{ \sum _{j}^{}\sqrt{c_{j}^{p}}} \left( 1- \lambda _{i} \right).
\end{equation}
Aggregating the expression in \eqref{deficitfoc} across all citizens yields
\begin{equation}
V^{p\prime} \left( F^{p} \right) - \Lambda =1-\frac{ \sum _{j}^{} \lambda _{j}\sqrt{c_{j}^{p}}}{ \sum _{j}^{}\sqrt{c_{j}^{p}}} \leftrightarrow V^{p\prime} \left( F^{p} \right) - \Lambda  \approx 1 \leftrightarrow V^{p\prime} \left( F^{p} \right)  \approx 1+ \Lambda.
\end{equation}
The approximation follows from the fact that  \(  \lambda _{i} \) is of order  \( \frac{1}{N} \).  In a large population the denominator in the square root sum ratio is much larger than the numerator. Thus, underfunding to good $p$,  when  \(  \lambda _{i} \) is of order \( \frac{1}{N} \) is on the order of \(1+\Lambda \). Underfunding is thus bounded by the sum of the shadow values $\lambda_i$ of reducing the deficit.


This analysis suggests that once we account for the deficit, the QF mechanism \textit{does not yield efficiency}. Instead it yields underfunding of all public goods by roughly  \( 1+ \Lambda  \). How to interpret this conclusion is somewhat subtle, and further analyses must be done to illustrate its consequences in a wider range of cases. In many cases, incorporating citizens' impact on the deficit may not fundamentally change our conclusions.  We now briefly run through some of these cases, though we acknowledge that many further analyses are required to make  general statements. In addition, experimentation is necessary to understand the settings in which these circumstances are more or less likely to obtain.

\begin{itemize}
	\item[\textbf{1.}] First consider the case in which most of the goods funded by the mechanism only benefit a relatively small fraction of the community and negative contributions are not allowed.  In this case, there is little or no problem, because our analysis relies on \textit{negative} contributions being made by all of the citizens that do not benefit (the left-hand side of the first-order condition in \eqref{deficitfoc} is negative). As long as negative contributions are disallowed (as in the baseline analysis in \ref{baseline}), most contributions will drop towards  \(\Lambda \) and we will obtain a conclusion very close to the one obtained by ignoring the financing considerations.  

	\item[\textbf{2.}]  In some cases with negative contributions, there will underfunding but it will still be, generically, an improvement over under private schemes because: (i) the magnitude of underfunding will often be moderate compared to purely private contributions and (ii) underfunding will be constant across different goods. We may want to allow for negative contributions because certain ``goods'' are public ``bads'' for some, such as financing hate speech.  Allowing such ``shorting'' may be undesirable in some cases, as we discuss in \ref{negative} below. Consider the setting in which some non-wasteful tax is used to fund the deficit so that \(  \Lambda =1 \). In this case, QF will lead to  \( V^{p\prime}=2 \), which is underfunding of public goods, but not severe underfunding relative to private contributory schemes. Importantly, the degree of underfunding is neutral across different goods and thus approximately optimal in the sense of Ramsey-Atkinson-Stiglitz taxation, as will be discussed in more depth in \ref{budgeted} below.  Furthermore, these considerations are entirely irrelevant for goods consumed by a small part of the population if we assume that most citizens will not be bothered to make tiny negative contributions to goods from which they do not benefit. Finally, this pessimistic conclusion can be overcome by reducing the cost of contributions to be proportionally smaller than the amount they influence outcomes.\footnote{If everyone is perfectly rational, this occurs in the extreme case in which a minuscule contribution affects funding by a large amount. We would not advocate this in practice as the risks of manipulation of such a system seem much worse than the underfunding by a factor of 2.}

	\item[\textbf{3.}] In cases with an external philanthropist funding the subsidies in the mechanism, as will be discussed in \ref{budgeted}, citizens are likely to ignore their impact on the deficit. When there is an external philanthropist, there is no need for a tax, and thus citizens are less likely to worry about their impact on the deficit in choosing their contribution levels.\footnote{It would be interesting, in future work, to compare the results under pure philanthropically-funded subsidies vs. taxes, as discussed in \citet{roberts1992}.} There will be some underfunding, but this is determined by the constraints of the philanthropist and not by a financing quirk in the mechanism.
\end{itemize}

In short, while incorporating citizens' impact on the deficit creates some complications and potential deviations from optimality, the impact may be small or irrelevant in particular cases of interest. Note, however, that the cases discussed here are not in any sense general. We have not discussed the precise formal conditions under which QF represents an improvement over purely private contributions, we have merely provided suggestive conditions that may or may not obtain in settings of interest. We have also not offered an analysis of when QF is superior to 1p1v when individuals take into account their impact on the deficit.  Furthermore, the cases require further analysis: case \textbf{2} (which allows negative contributions) and case \textbf{3} (which relies on philanthropic funds) both present a host of other potential issues that will be discussed in more detail though not completely resolved in \ref{negative} and \ref{budgeted}, respectively.

\section{Variations and Extensions}\label{extensions}

The above sketch leaves us with many open questions.  In this section we address a handful of the most critical outstanding questions about the mechanism, and highlight particularly important directions for future analysis.

\subsection{Budgeted matching funds}\label{budgeted}


In many practical applications the funding for QF is likely to come from philanthropists or some dedicated government appropriation rather than from unlimited tax revenue. There are clear theoretical advantages of such philanthropic (or dedicated government) funding. If participants do not personally care about the philanthropist's wealth, the issues around incorporating impacts on the deficit (as discussed in \ref{deficit} above) disappear. However, even the wealthiest philanthropists do not have infinite funds and thus cannot simply agree to finance arbitrarily large deficits.  In this subsection we describe a variant on the QF mechanism that can limit the funding required. 

Consider a rule that is an  \( \alpha \) mixture of QF with a  \( 1- \alpha  \) weight on un-matched private contributions. We call this the Capital-constrained Quadratic Finance (CQF) mechanism. 
\begin{definition}[Capital-constrained Quadratic Finance Mechanism] The Capital-constrained Quadratic Finance Mechanism $\Phi^{CQF}$ satisfies
 \[ F^{p}= \Phi^{CQF}(\cdot) = \alpha  \left(  \sum _{i}^{}\sqrt{c_{i}^{p}} \right) ^{2}+ \left( 1- \alpha  \right)  \sum _{i}^{}c_{i}^{p}. \]
\end{definition}
The first feature to note about $\Phi^{CQF}$ is that for any budget  \( B \),  \(\alpha\) may be adjusted to ensure the budget is not exceeded.  To see this, note that when  \(  \alpha  \rightarrow 0 \) the mechanism is both directly self-financing and, indirectly, the amount invested in the public good falls for the reasons we have discussed above.  Thus, the deficit can be eliminated by setting  \(\alpha\) low enough.  This flexibility ensures that a philanthropist can reliably set a low level of  \(\alpha\) and perhaps gradually increase it over time to increase support.

Second, note that no one would ever choose to contribute outside this system (no one's contributions through it are taxed), so CQF is individually rational. 
 \begin{proposition} The mechanism $\Phi^{CQF}$ is individually rational in the sense that $\frac{ \partial \Phi^{CQF}}{ \partial c_{i}^{p}} > \frac{ \partial \Phi^{priv}}{ \partial c_{i}^{p}}$.
 \end{proposition}
 \begin{proof}
To show that CQF is individually rational, compare the marginal impact of individual $i$'s a contribution to good $p$ through $\Phi^{CQF}$ to the marginal impact of her contribution to $p$ through a separate mechanism $\Phi^{priv}$. We showed in \ref{model} that the marginal value of a contribution under $\Phi^{priv}$  is equal to 1. Consider the marginal contribution under $\Phi^{CQF}$: 
\begin{equation}
\frac{ \partial F^{p}}{ \partial c_{i}^{p}}= \alpha \frac{ \sum _{j}^{}\sqrt{c_{j}^{p}}}{\sqrt{c_{i}^{p}}}+1- \alpha. 
\end{equation} 
The factor multiplying  \(  \alpha  \) is by construction always at least  \( 1 \), so this always exceeds unity, the marginal impact of a contribution made through a separate, purely private channel.  
\end{proof}

The individual rationality property suggests that CQF is consistent with existing within a broader society in which private contributory schemes are the norm, not just in terms of funding but also in terms of getting citizens to ``play ball" with the mechanism.

Third, consider equilibrium incentives under CQF.  In choosing her contribution to good  \( p \), citizen  \( i \) maximizes
\begin{equation}\label{clrmax}
 V_{i}^{p} \left(  \alpha  \left(  \sum _{j}^{}\sqrt{c_{j}^{p}} \right) ^{2}+ \left( 1- \alpha  \right)  \sum _{j}^{}c_{j}^{p} \right) -c_{i}^{p}.
\end{equation}
\begin{proposition} If the population $M$ funding good $p$ is large relative to any individual contribution $c_i^p$, then $\Phi^{CQF}$ leads to underfunding relative to purely private contributions. Underfunding for good $p$ is $\frac{1}{\alpha}.$ When $\frac{1}{\alpha} \ll M$, $\Phi^{CQF}$ yields less underfunding than $\Phi^{priv}$.
\end{proposition}
\begin{proof}
The first order condition of (\ref{clrmax}) is
\begin{equation}
 V_{i}^{p'} \left[  \alpha \frac{ \sum _{j}^{}\sqrt{c_{j}^{p}}}{\sqrt{c_{i}^{p}}}+1- \alpha  \right] =1 \leftrightarrow V_{i}^{p\prime} \approx \frac{\sqrt{c_{i}^{p}}}{ \alpha  \sum _{j}^{}\sqrt{c_{j}^{p}}} \leftrightarrow V^{p\prime}=\frac{1}{ \alpha }. 
\end{equation}
The approximation comes from the fact that $c_i^p \ll M$. 
\end{proof}

The approximation requires that the population funding the good is large relative to any individual. This approximation is natural for a genuinely public good; for goods supplied to very small communities or citizens, funding will be greater than implied by this approximation, but this extra funding will mostly come through the private channel and not be subsidized by the philanthropist and thus should not be of great concern to her.  

Thus, the CQF mechanism will lead to underfunding of the good by a factor of  \( \frac{1}{ \alpha } \) as compared to the (rough) underfunding under purely private contributions by a factor of the typical size of the benefiting community. Assuming  \( \frac{1}{ \alpha } \) is small relative to  \( M \), CQF can dramatically improve funding relative to purely private contributions. 

Furthermore, subject to the budget constraint, funding is approximately optimally allocated across different public goods in the sense of \citet{ramsey} taxation and the important extension to allow for heterogenous consumers by \citet{atkinsonstiglitz}. The basic idea of Atkinson-Stiglitz taxation is that, when considering commodity taxation, it is optimal to distort the consumption of all goods equally, so that the marginal rate of substitution across all goods is the same. To see how the Atkinson-Stiglitz logic applies in our setting, we consider the planner's problem which is the same as in the baseline set up, but with a new interpretation of the budget constraint. The planner seeks to maximize  \(  \sum _{i}^{}V_{i} \left( F^{p} \right)  \) subject to the budget constraint which is simply  \(  \sum _{p}^{}F^{p}=B \). Solving the constrained maximization problem, 
 \begin{equation}
 \sum _{i}^{}V_{i} \left( F^{p} \right) - \lambda  \left(  \sum _{p}^{}F^{p}-B \right), 
\end{equation}  
gives \( V^{p\prime}= \lambda  \), i.e.  \( V^{p\prime} \) is a constant. Thus our above result that  \( V^{p\prime}=\frac{1}{ \alpha } \)  suggests that CQF funding is optimally allocated across goods if  \(  \alpha  \) is chosen to just exhaust the budget.

Of course, this analysis ignores the fact that funding different goods differentially may help stimulate more private contributions. We also ignore the fact that CQF does not quite achieve  \( V^{p\prime}=\frac{1}{ \alpha } \) as there are also some contributions through the private channel. \citeauthor{atkinsonstiglitz}'s analysis is much more careful on these points and gives (fairly specific) conditions under which equal distortion ratios are nonetheless optimal.  Verifying conditions when CQF is exactly optimal is an interesting direction for future research, but is beyond the scope of this paper.

Some of the underfunding implied by CQF may not be entirely undesirable. This underfunding may balance under-investment in private goods creation required by the distortionary taxes that will often be necessary to fund the mechanism. As we now discuss, the underfunding implied by CQF may also help to deter collusion.

\subsection{Collusion and fraud}\label{collusion}

The central vulnerabilities of QF are collusion and fraud. These vulnerabilities are common to most other mechanisms designed based on the assumption of unilateral optimization. Collusion takes place when multiple agents act in their mutual interest to the detriment of other participants. Fraud takes place when a single citizen misrepresents herself as many.

It is useful to spell out precisely what these threats are and the harms they could bring to QF or CQF.  Consider, for concreteness, a case of CQF with  \(  \alpha =.1 \). First suppose one citizen is able to misrepresent herself fraudulently as 20.  If she contributes  \( x \) dollars in the capacity of each of these citizens, she will pay  \( 20x \) but her cause (which could just deposit to her bank account) will receive
 \[ .1\cdot \left( 20\sqrt{x} \right) ^{2}=40x. \] 
Thus, on net, she doubles her money. This is a sure arbitrage opportunity and could easily convert QF into a channel for lining the pockets of the fraudster.  The minimum fraud size required to run this racket at positive profit is  \( \frac{1}{ \alpha }. \)

A perfectly colluding group of citizens could achieve something similar. The colluding group may all be participants in the mechanism, or they may be partially formed of participants in the mechanism together with one or more outside observers with an interest in the mechanism's outcome. Collusion can either happen ``horizontally'', between multiple participants with similar goals, or ``vertically'', between one or more participants in the mechanism and an outside participant (or a participant in a different side of the mechanism, e.g. a potential recipient of a subsidy) that can offer conditional payments (i.e. bribes) to induce the participants to behave in particular ways. Again, if the size of this group is greater than  \( \frac{1}{ \alpha } \) and the group can perfectly coordinate, there is no limit (other than the budget) to how much it can steal.

However, note that unilateral incentives run quite strongly against certain forms of collusion.  Consider a colluding group with  \( 100 \) members each investing \$1000, which is thus funded at a level of  \( .1\cdot 100^{2}\cdot 1000=\)\$1,000,000. If this cartel divides the spoils equally among its members, the group members each receive \$10,000 and thus achieve a net benefit of \$9000.  Now consider what happens if one member decides to defect and contribute nothing.  The funding level is now  \( 99^{2}\cdot 100= \)\$980,100. The defecting member would see her pay out fall to \$9801, but would have saved \$1000 and thus on net would now be making \$801 more than she was before.  There is thus very little incentive for any member of the cartel to actually participate. Unless activity can be carefully monitored and actual payment levels directly punished, defection is likely to be very attractive and the cartel is likely to die the death of a thousand cuts.  Simply sharing revenue with participants is not sufficient to sustain collusion.

There is a broader point here.  If perfect harmonization of interests is possible, purely private contributions lead to optimal outcomes. QF is intended to overcome such lack of harmonization and falls prey to manipulation when it wrongly assumes harmonization is difficult. So we're led into a bit of a paradox: QF seeks to foster community direction through its design, but in doing so QF relies on the strong ties of community flowing outside the design not existing.

The appropriate way of deterring fraud and collusion will depend on the affordances of the system.  First consider fraud, which is the simpler and more devastating of the issues.  If fraud cannot be reasonably controlled, QF simply cannot get off the ground; it will immediately become a money pump for the first fraudster to come along.  Note, however, that this is true of nearly any system with a democratic flavor: 1p1v can easily be exploited through fraud.  The simplest and most clearly necessary solution to fraud is an effective system of identity verification.  Beyond identity verification, relatively small groups giving large contributions and thus receiving large funding should be audited when possible to determine if fraud has occurred and large penalties (much larger than the scale of the fraud, to adjust for the chance of detection) should be imposed on the fraudsters and transferred to other, honest citizens.

Collusion is a subtler and more pernicious problem to root out, and perhaps the greatest challenge for QF, given the tension between community building and collusion deterrence.  In all cases, a modest value of  \(  \alpha  \) and auditing of small, highly funded groups will help deter tight collusive groups. Yet the best approach to deterring broader collusion will depend on the nature of the setting: a case in which citizens are friendly and all know each other, as in a small town, will differ from the case in which participants have low trust for each other and are highly diverse, as in a blockchain community.

\subsection{Negative contributions}\label{negative}

Not all public projects bring benefits alone; some may harm certain citizens by creating negative externalities such as pollution or offense. The existence of negative externalities does not immediately imply we should allow negative contributions to reflect these harms. Some of these negative externalities can be addressed directly through legislation. Furthermore, there are dangers of allowing citizens to defund projects they don't like.  Allowing negative contributions opens many thorny issues as we discussed in \ref{deficit} above. However, in some cases the benefits of allowing the expression of negative externalities will outweigh the potential costs and thus negative contributions will be desirable. While further theoretical analyses are required to fully understand the conditions under which negative contributions would be desirable, we provide a brief discussion here of how allowing negative contributions changes the baseline analysis presented in \ref{baseline}.

The natural extension of QF to allow negative contributions is one in which citizens may choose to defund a public good according to the same cost structure.  
\begin{definition}[$\pm$ Quadratic Finance Mechanism] The $\pm$ Quadratic Finance Mechanism, $\Phi^{\pm QF}(\cdot)$ satisfies \[ F^{p}= \Phi^{\pm QF}\left(  \sum _{i}^{} \pm _{i}\sqrt{c_{i}^{p}} \right) ^{2}, \]
where  \(\pm _{i}\) is positive or negative at the discretion of citizen  \( i \).
\end{definition}
Citizens with  \( V_{i}^{p'} \geq 0 \) or  \(  \lambda _{i} \) in the cases where they account for their budget impact will choose the positive sign; those with the opposite will choose the negative sign.  

We already know the first-order condition for positive contributors; let's consider it for negative contributors (for simplicity we focus on the deficit-ignoring, fully financed case):
\begin{equation}\label{pmlr}
 -\frac{V_{i}^{p\prime} \left( F^{p} \right)  \left(  \sum _{j}^{} \pm _{j}\sqrt{c_{j}^{p}} \right) }{\sqrt{c_{i}^{p}}}=1 \leftrightarrow V_{i}^{p\prime} \left( F^{p} \right) =-\frac{\sqrt{c_{i}^{p}}}{ \sum _{j}^{} \pm _{j}\sqrt{c_{j}^{p}}}. 
\end{equation}
Note that together with the first-order condition for those making positive contributions, \eqref{lrcontributions} and \eqref{pmlr} can be summarized as
\begin{equation}
 V_{i}^{p\prime} \left( F^{p} \right) =\frac{ \pm _{i}\sqrt{c_{i}^{p}}}{ \sum _{j}^{} \pm _{j}\sqrt{c_{j}^{p}}}. 
\end{equation}

Aggregating across citizens yields \( V^{p\prime} \left( F^{p} \right) =1,\) as is optimal.  Allowing negative contributions is thus desirable in the following sense: without allowing them, there may be negative externalities of a project that are not internalized into its funding. 

However, as noted in \ref{deficit} above, it is principally allowing negative contributions that leads to underfunding if citizens consider their impact on the deficit.  More broadly, negative contributions may be a quite powerful way to deter collusive schemes as they offer a way for any citizen to be a ``vigilante enforcer'' against fraud and abuse.  The downside of this benefit, however, is obviously that, in some cases, absolute free speech and other protections may lead us to distrust such vigilantism.  

In short, there are a variety of costs and benefits to allowing negative contributions and we suspect their desirability will vary across contexts.

\subsection{Variations on functional form}\label{form}

One might naturally wonder if the functional form we propose is uniquely optimal.\footnote{As \citet{eguiaxefteris} show, in a sufficiently large population (holding fixed value distributions), any function with a zero first derivative and positive second derivative will behave like QV; this idea likely extends to our present context. However, this result may be of limited relevance in our setting---the appropriate limit is often one where value distributions also change as the population grows, so that behavior of the function away from zero also matters.} We leave a formal proof for future work. Here, we plumb intuition by considering a class of rules that nests both QF and purely private contributions. 

Consider rules $\Phi^\beta$ that satisfy
\begin{equation}
  F^{p}=\Phi^{\beta}(\cdot) = \left(  \sum _{i}^{} \left( c_{i}^{p} \right) ^{\frac{1}{ \beta }} \right) ^{ \beta }. 
\end{equation}
Again, we analyze $\Phi^{\beta}$ abstracting from deficits and incentives created by mechanisms that take this form. To avoid redundancy, we skip straight to citizen  \( i \)'s first-order condition:
\begin{equation}
\frac{V_{i}^{p\prime} \left(\sum _{j}^{} \left( c_{j}^{p} \right) ^{\frac{1}{ \beta }} \right) ^{ \beta -1}}{ \left( c_{i}^{p} \right) ^{\frac{ \beta -1 }{ \beta }}}=1 \leftrightarrow V_{i}^{p\prime}=\frac{ \left( c_{i}^{p} \right) ^{\frac{ \beta -1}{ \beta }}}{ \left(  \sum _{j}^{} \left( c_{j}^{p} \right) ^{\frac{1}{ \beta }} \right) ^{ \beta -1}} \leftrightarrow V^{p\prime}=\frac{ \sum _{j}^{} \left( c_{j}^{p} \right) ^{\frac{ \beta -1}{ \beta }}}{ \left(  \sum _{j}^{} \left( c_{j}^{p} \right) ^{\frac{1}{ \beta }} \right) ^{ \beta -1}}. 
\end{equation}

A convenient property of this form is that for \( \beta >1\), unlike for the  \(  \beta =1 \) case, every citizen with strictly positive  \( V_{i}^{p\prime} \) will make a positive contribution.  Note that as  \(  \beta  \rightarrow 1 \) however, this rule approaches purely private contributions, while when  \(  \beta  \rightarrow 2 \) the rule becomes QF. 

Away from these now-familiar cases, it is useful to consider what happens for  \(\beta  \in  \left(1,2 \right)\) and  \(  \beta  \in  \left( 2,\infty \right)  \).\  Note that our reasoning above implies that  \( V^{p\prime} \) is in all these cases equated to something of the form
\begin{equation}
\frac{ \sum _{i}^{}h \left( x_{i} \right) }{h \left(  \sum _{i}^{}x_{i} \right) }, 
\end{equation}
where  \( x_{i} \equiv  \left( c_{j}^{p} \right) ^{\frac{1}{ \beta }} \) and  \( h \left( x \right)  \equiv x^{ \beta -1} \).
Whether this ratio is greater than or less than one is determined by Jensen's inequality.
That is, public goods will be over (under) funded if the function  \( x^{ \beta -1} \) is convex (concave). Given that  \(  \beta =2 \) leads to efficiency and  \(  \beta =1 \) leads to the severe underfunding of purely private contributions, this result should not be too surprising.\par

Might  \(  \beta  \in  \left( 1,2 \right)  \) be a superior interpolation between purely private contributions and QF when compared to our CQF mechanism in \ref{budgeted} above?  While such solutions are worthy of experimentation, theory indicates their inferiority.  To see why, note that  \(  \beta  \in  \left( 1,2 \right)  \) does not simply lead to underfunding, but to \textit{differential} underfunding of projects with many small contributors.  To see this note that we can rewrite citizen \(i\)'s first-order condition as
\begin{equation}
\left( V_{i}^{p\prime} \right) ^{\frac{1}{ \beta -1}}=\frac{ \left( c_{i}^{p} \right) ^{\frac{1}{ \beta }}}{ \sum _{j}^{} \left( c_{j}^{p} \right) ^{\frac{1}{ \beta }}} \leftrightarrow  \sum _{j}^{} \left( V_{j}^{p\prime} \right) ^{\frac{1}{ \beta -1}}=1. 
\end{equation}
Thus the efficiency condition that the aggregate marginal utility equals one obtains \textit{except} that the transformation  \( x^{\frac{1}{ \beta -1}} \)\ is applied to it.  For  \(  \beta <2 \) this transformation is convex, which will thus exaggerate large marginal utilities and dampen small ones. Thus, $\beta <2$ systematically leads to the underfunding of goods with many small beneficiaries and over-funding of goods with a few large beneficiaries. 

This result may be problematic for two reasons. First, it is problematic from an efficiency standpoint.  It is worse than the budget-constrained efficiency we (approximately) obtained in \ref{budgeted} above from CQF. In addition, it would seem to make small group collusion quite profitable. 

We do not mean to suggest that using a function other than quadratic has no purpose.  It may be useful, in some cases, to replace the square root and square functions with ones that behave more like the absolute value near the origin and only become quadratic further out to avoid large groups engaging in collusion. If a group of individuals colluded and each contributed very small amounts of money each, they could run a highly profitable scheme for very little cost. And CQF does relatively overfund goods with intense supporters, though the extent to which it does so is modulated by $\alpha$. Generally, we view these other functional forms primarily as a foil that helps us understand QF and the failures of private contributory schemes, rather than as viable alternative funding mechanisms.

\subsection{Failures of concavity and dynamic solutions}\label{concavity}

Above we assumed that all functions  \( V_{i}^{p}\) were smooth and concave.  Further analysis is required to fully understand the limitations of this assumption---indeed further analyses may point to alternative ways of structuring the mechanism. However, for now we illustrate that this assumption is innocuous in many circumstances. 

Consider, for example, the case in which the value derived from a public good is S-shaped (sinusoidal). Unless the good is funded ``sufficiently" citizens derive little value from it. Once it is funded sufficiently the marginal value of funding quickly diminishes.  This is a natural structure for projects with a nearly-fixed budget, such as public infrastructure projects.  In this case, citizens will not be willing to contribute unless they expect others to do so as well.

A natural solution to this problem is what is often called an ``assurance contract'' and was proposed by \citet{dybvig}. The most natural implementation is dynamic, and we suspect this is how an QF mechanism would operate in practice in any case, but static implementations are also possible.\footnote{For example, citizens could state a schedule of how much they would like to contribute, conditional on the contributions of others, or some coarse approximation thereof, such as a minimum threshold for their contribution.  An automated system could then calculate an equilibrium of these requests.} Essentially, there is a window of time during which contributions and withdrawals of contributions to the mechanism are made.  Citizens are thus able to contribute without fear that they will be left ``exposed'' to the risk that others will not contribute.  Given this, every citizen may as well make a reasonable contribution until the relevant threshold has been reached.  In the spirit of \citet{tabarrok}, an entrepreneur confident that a good is worth funding can further sweeten the deal by offering citizens a payment if they agree to temporarily fund the mechanism to avoid any potential coordination problems.\footnote{Note, however, that Tabarrok's suggestion that such a scheme alone is enough to fund public goods is problematic: it is based on a non-generic assumption of precisely infinite derivative in the utility functions at a single point. Any sinusoidal structure that is smooth will destroy the result and lead to arbitrary underfunding, whatever the ``assurance'' structure. The same logic applies to proposals that draw on discontinuous payoff functions for generic informational structures on values \citep{postelwaite}. }

This dynamic implementation is very likely to be desirable even if  \( V_{i}^{p} \) is concave, as the optimal contribution will still depend on others' contributions.  Thus, it makes little practical difference whether the value functions are concave, except possibly for the chance of a weak cold start problem, which an assurance contract scheme, \`{a} la Tabarrok, could address. Similar resolutions may apply to cases in which smoothness fails. 

\section{Applications}\label{apps}

We discuss several applications of QF in order to illustrate the importance of its many nice features in practice. We focus on an application to campaign finance reform. Then, we briefly discuss a few other potential applications to illustrate the range of settings, across quite distinct domains, in which we believe QF can be implemented in the relatively near term.

\subsection{Campaign finance}

In the US, the regulation of individual and collective contributions to political campaigns has been hotly debated since the first attempts to regulate campaign finance in the mid-1800s. The 1971 Federal Election Campaigns Act and subsequent amendments introduced extensive rules and procedures for campaign funding geared toward balancing transparency and equity with freedom of expression, and established the Federal Elections Commission (FEC) to regulate the fundraising activities of candidates for public office. Campaign finance issues frequently make their way to the Supreme Court---and the court's Citizens United decision has maintained a steady stream of vigorous opposition since its ruling in 2010.

The proposals for campaign finance reform are manifold. Suggesting modifications to municipal, state, and federal election law, these proposals range from simple tweaks of existing laws (e.g. capping contributions, stricter enforcement, restricting contributions from unions and corporations, etc.) to extensive re-envisioning of electoral systems (e.g. public financing schemes, anonymous capped contributions, etc.).\footnote{For influential discussions of campaign finance reform, see \citet{aa}, \citet{lessig} and \citet{hasen}.} The proposals for reform offer solutions to the core legal and political question: How can regulatory bodies strike a balance between freedom of expression through contributions to campaigns for elected office, while restricting the undue influence of special interests?   

The motivating problem for campaign finance reform can be analyzed using the formal apparatus presented in previous sections. When un-checked, permissive campaign finance laws such as the ones upheld in Citizens United are purely private contributory schemes. As demonstrated above, the private contributions mechanism for flexible funding of public goods leads to tyranny of the few who have resources to make very large contributions. In the campaign finance setting, the failure of a purely private contributory mechanism implies that on the margin, only a single contributor (the largest contributor) has any influence. The motivating problem of campaign finance is existing systems' vulnerability to tyranny of the rich, especially when one considers the possibility for \textit{quid pro quo} corruption. Just as the QF mechanism answers to the central problem of purely private funding, it provides a template for a new proposal for campaign finance reform.

The QF mechanism solves the funding problems with existing systems by boosting the contributions of small donors, thereby effectively diluting the influence of larger ones. Under existing schemes, individuals able to make only small contributions have little incentive to contribute, knowing that their contributions are just a drop in the bucket. Under QF-based campaign finance, all individuals have incentive to contribute as long as their evaluation of the candidate is positive. This fact also has good second-order outcomes that are the converse of \textit{quid pro quo} corruption under purely private contributions---since all individuals have incentive to contribute, campaigning politicians thus have to give some weight to every individual in their electorate. Under QF-based campaign finance, fundraising and outreach are intertwined, leading politicians to engage more thoroughly and deeply with their electorate.

The rationale for moving toward QF in campaign finance in part parallels the rationale behind existing political matching funds. Public matching funds for campaigns---such as the federal  matching fund for presidential elections\footnote{The federal matching fund for presidential campaigns is financed by a \$3 voluntary contribution on income tax returns.} and municipal and state matching funds for legislatures and mayoral races\footnote{As mentioned in the introduction, New York City's matching fund policy has led other cities and states to consider similar procedures. However, several states had publicly financed matching funds deemed unconstitutional by the 2011 Supreme Court Cases,  \textit{Arizona Free Enterprise Fund v. Bennett } and\textit{ McComish v. Bennett}.}---aim to amplify small contributions to campaigns for elected office. Thus, like the QF mechanism, matching funds subsidize small contributions. Yet existing matching funds systems are highly arbitrary: often they match  contributions linearly with some chosen scaling factor, and up to some selected level. How are the maximums and the match ratios chosen?  Shouldn't there be a more gradual taper of matching commitments? QF gives an optimal mechanism for achieving a flexible matching fund, with clear logic behind the match ratio for any given contribution.

The usual rationale behind matching systems, at least as described to the public, is some pretense of equity. Meanwhile, the rationale for an QF-based system does not even rely on an argument from equity. Indeed, QF is an (approximately) optimal mechanism from an efficiency perspective as we showed in \ref{budgeted} above. QF for campaign finance is thus joined to matching funds proposals in spirit, while offering substantial improvements over these existing proposals in practice. 

The efficiency rationale for QF and the fact that it does not tax speech has another important practical benefit: it would likely pass constitutional muster even under the post-Citizens United state of American constitutional law.  Large contributions are not taxed as we showed in \ref{budgeted}.  Thus, the system simply boosts the relative importance of small donors.

The design details of an implementation of QF in the context of local, state or federal campaigns require attention to many subtleties. Should there continue to be contribution limits until enough funds can be raised to make \( \alpha\) reasonably high?  What should a campaign have to do to be allowed to list in the system?  Should different candidates be allowed to form ``parties'' then disperse funds to their candidates?  Should contributions be public or doubly blind (as in Chile) so that candidates and parties do not know their own contributors, thus reducing corruption and collusion? Though many questions remain, overall, we believe the structure of QF/CQF could simplify the byzantine patchwork of current campaign finance regulations.

\subsection{Other applications}

To gesture to the range of other possible applications of QF/CQF, we briefly review four other promising domains: open source software, news media finance, charitable giving and urban public projects. 

\begin{itemize}

\item[\textbf{1.}] The open source software movement is based on the principle that code is or should be a public good.  Software is a classic example of an increasing returns activity, as it is nearly costless to copy and apply broadly, yet has potentially large upfront development costs, especially when the uncertainty of any solution working out is factored in.  Many in the software community view exclusionary capitalist solutions as wildly inefficient and undesirable, yet democratic or government-driven provision is usually far too hierarchical and centralized for fast moving technology appreciated primarily at first by a small community.\footnote{For a classic exposition of dynamics of open source development, see \citet{benkler}.} QF could provide a flexible solution that allows for public and private contributions to open source development, without compromising the principles upon which the movement was founded.\footnote{The organization Gitcoin used the CQF mechanism to dedicate \$25,000 to open source software projects in February 2019 (\url{https://medium.com/gitcoin/radical-results-gitcoins-25k-match-2c648bff7b19}).}

\item[\textbf{2.}] Financing the production of news is an especially fitting application of the QF mechanism. On the one hand, news (especially high-quality, investigative journalism) is perhaps the clearest example of a public good.  It can be costly to create, but it is essentially impossible to exclude anyone from consuming it beyond a very tiny window of time and thus it is very difficult to earn value without highly costly and wasteful mechanisms of exclusion.\footnote{ This problem has become increasingly acute with the rise of information and communications technology, leading to an increasing sense of crisis in the funding of news, which some have even labeled as ``existential'' \citep{foer}.} Yet news is also often relevant to a very broad community, making purely charitable funding difficult to pull off. This creates a strong desire for public funding and is the reason that governments all over the world are involved in news production.  However, the drawbacks of government involvement in news creation could hardly be more evident, given the central role of media in holding governments to account. QF offers a potentially appealing resolution.  Governments and philanthropists interested in supporting high-quality news without exerting or being seen to exert undue influence over content could use QF to effectively match donations to news creators in much the way that they already match contributions to organizations like National Public Radio in the United States.  Using CQF rather than standard matching would create greater efficiency and would require less targeted and discretionary applications of funds, thereby allowing a truly diverse ecosystem of news outlets to flourish.

\item[\textbf{3.}] QF also aligns well with existing movements toward more democratic forms of charitable giving. Many philanthropists provide matches to favorite charities and many are seeking more creative ways to harness decentralized information outside the philanthropist's whims to give away money. The Open Philanthropy and ``effective altruism'' movements are based on the idea that donor discretion should be removed from philanthropy to the greatest extent possible.  In areas where randomized controlled trials and other precise measurements are insufficient to direct funds, QF seems well-suited. Across a wide range of domains, from funding educational start-ups to large scale interventions in developing countries, CQF holds the potential to provide more accurate and less hierarchical signals for directing charitable funds.  Such non-hierarchical mechanisms for charitable giving are increasingly relevant as backlash continues to grow against the top-down dictates of well-intentioned but ultimately elitist class of donors \citep{easterly, giriharadas, reich}.\footnote{In fact, one  organization, WeTrust, has already run a charity donation matching campaign for 501(c)3 nonprofits using CQF (\url{https://www.ethnews.com/wetrust-experiments-with-liberal-radical-donation-matching}).}

\item[\textbf{4.}] While urbanists have long recognized the importance of community-level decision-making in cities, cities often lack mechanisms that allow goods valued in communities to emerge. QF, as applied to urban public funding decisions, could allow communities at all scales to fund projects that would struggle to get funding under centralized systems. A growing body of evidence suggests that policies emphasizing community values and diversity generate major improvements in city life. But city councils and other municipal governments struggle to meet the needs of sub-communities. Even though they are democratic systems intended to represent the will of a constituency, the needs of very small groups cannot be heard for reasons we have discussed. Some public goods are intensely important to a select few, for example a small group of households clustered in a few city blocks. And yet, the systems in place for communicating those needs and receiving the adequate funding are highly inefficient. QF, as applied to city planning, aligns nicely with the ideas advanced by some of the most prominent modern urban theorists.\footnote{Activist and intellectual Jane \citet{jacobs} famously condemned the urban planning ethos of her time, arguing that ``rationalist" urban planners do a poor job serving the needs of actual city-dwellers, undoing the sense of community that makes people move to cities in the first place through their top-down, deductive approach to allocation and decision-making. Similarly, anthropologist and geographer David Harvey has long recognized the importance of the city as a locus of self-definition through community attachment \citep{harvey}. Harvey emphasizes that precisely because the urbanization process creates so much surplus, the ``right to city" demands new forms of democratic management of that surplus.} The city is a fertile site for application of QV, and these applications align with the logic long-advanced by political economists and urban theorists alike.

\end{itemize}

\section{Conclusion}\label{conclusion}

In this paper, we presented a novel funding mechanism for an emergent ecosystem public goods, and highlighted key areas for further analysis, experimentation and application. Our treatment of the economic theory around this mechanism was superficial. Our analysis is based on simplified assumptions of quasi-linear utility and independent values. There is significant room for improvement in the analysis of the financing mechanism and deficits' influence on incentives. Our discussion of collusion is thin. Moreover, we do not even touch distributive issues. Experiments can help guide these theoretical analyses---it would be especially valuable to see experiments that compare QF to related mechanisms such as VCG.

Beyond economic theory, there are countless implementation questions our discussion leaves open. QF's formal structure will initially strike many as bizarre.  Designing interfaces, helping participants to ``see'' how QF works would require educating citizens and planners prior to implementation. Is it possible to make wide publics largely unfamiliar with mathematics comfortable with the QF mechanism? What can be done to further defend QF from attacks and hacks? Would QF optimally encourage community formation or would the norms and rules needed to avoid collusion inadvertently undermine important communities? Further, it is critical to understand how participants dynamically assess the marginal benefits they derive from different types of goods at different funding levels. Psychological constraints may suggest particular domains of application over others. 


We hope for a wide range of experimentation around QF---thus we have laid out a variety of narrower domains in which experimentation has already begun or otherwise may be most plausible in the near term.  Such experimentation is critical for a variety of reasons: to investigate the weaknesses of the formal mechanisms and address these flaws with new designs, to acquaint people with their operation and build awareness of their value, to build social institutions around them that make them effective, and to provide rigorous empirical evidence of their value.

\bibliographystyle{aer}

\bibliography{LRbib}
\end{document}